\documentclass[conference,onecolumn]{IEEEtran}

\usepackage[english]{babel}
\usepackage{enumitem}
\usepackage[latin1]{inputenc}
\usepackage{enumerate}
\usepackage{graphicx}
\usepackage{color}
\usepackage[dvipsnames]{xcolor}
\usepackage[T1]{fontenc}
\usepackage{subfigure}
\usepackage{dsfont}
\usepackage{amsfonts}
\usepackage{graphicx,wrapfig}
\usepackage[T1]{fontenc}
\usepackage{amsmath}
\usepackage{mathtools, cuted}
\usepackage{amsthm}
\usepackage{amstext}
\usepackage{amssymb}
\usepackage{mathrsfs}
\usepackage[square,numbers, sort&compress]{natbib}
\usepackage{mathtools}
\usepackage{tikz}
\usepackage{marginnote}
\usepackage{tkz-tab}
\usetikzlibrary{topaths,calc}
\usepackage{pgfplots}
\usepackage{stackengine}
\usepackage{algorithmic}
\usepackage{algorithm}

\usepackage[colorinlistoftodos]{todonotes}

\usetikzlibrary{arrows,positioning, shapes}
\usepackage{float}

\usepackage{setspace}
\usepackage{hyperref}
\hypersetup{colorlinks,
   colorlinks,
    citecolor=red,
    filecolor=green,
  linkcolor=blue,
    linktoc=page,
   urlcolor=blue
}

\def\R{\mathbb{R}}

\def\eps{\varepsilon}

\def\E{{\mathbb E}}

\def\P{{\mathcal P}}

\def\X{{\mathcal X}}
\def\Y{{\mathcal Y}}
\def\M{{\mathcal M}}
\def\N{{\mathcal N}}

\def\L{{\mathcal L}}

\def\sQ{{\mathsf Q}}
\def\sK{{\mathsf K}}

\def\sEs{{\mathsf E_{e^{\varepsilon}}}}
\def\sE{{\mathsf E}}

\usepackage[font=small,labelsep=space]{caption}
\DeclareCaptionLabelSeparator{dot}{.~}
\newcommand{\al}[1]{\begin{align*}
#1
\end{align*}}

\newcounter{example}
\newenvironment{example}[1][]{\refstepcounter{example}\par\medskip
   \noindent \textit{Example~\theexample. #1} \rmfamily}{\medskip}
\newtheorem{definition}{Definition}
\newtheorem{theorem}{Theorem}

\newtheorem{proposition}{Proposition}
\newtheorem{lemma}{Lemma}

\usepgfplotslibrary{fillbetween}
\usepackage{times}
\usepackage{tikz}
\usepackage{amsmath}
\usepackage{verbatim}
\usetikzlibrary{arrows,shapes}
\tikzstyle{RectObject}=[rectangle,fill=white,draw,line width=0.2mm]
\tikzstyle{line}=[draw]
\tikzstyle{arrow}=[draw, -latex]
\usetikzlibrary{decorations.pathmorphing}
\usetikzlibrary{calc,shapes, positioning}
\usepackage{graphicx}
\usetikzlibrary{shapes.geometric}

\DeclareFontFamily{U}{BOONDOX-calo}{\skewchar\font=45 }
\DeclareFontShape{U}{BOONDOX-calo}{m}{n}{
	<-> s*[1.05] BOONDOX-r-calo}{}
\DeclareFontShape{U}{BOONDOX-calo}{b}{n}{
	<-> s*[1.05] BOONDOX-b-calo}{}
\DeclareMathAlphabet{\mathcalboondox}{U}{BOONDOX-calo}{m}{n}
\SetMathAlphabet{\mathcalboondox}{bold}{U}{BOONDOX-calo}{b}{n}
\DeclareMathAlphabet{\mathbcalboondox}{U}{BOONDOX-calo}{b}{n}

\definecolor{DukeBlue}{HTML}{001A57}
\definecolor{DarkRed}{rgb}{0.75, 0.0, 0.0}
\definecolor{DarkGreen}{rgb}{0.0, 0.5, 0.0}

\allowdisplaybreaks

\author{}
\date{}

\begin{document}

	\title{\vspace{5.5mm}Privacy Amplification of Iterative Algorithms via Contraction Coefficients}

\author{%
Shahab Asoodeh${}^\dagger$, Mario Diaz${}^*$, and Flavio P. Calmon${}^\dagger$ \\
\small ${}^\dagger$Harvard University, \{shahab, flavio\}@seas.harvard.edu \\
$^*$Universidad Nacional Aut\'{o}noma de M\'{e}xico, mario.diaz@sigma.iimas.unam.mx
}
	
	\maketitle
\begin{abstract}
	We investigate the framework of privacy amplification by iteration, recently proposed by Feldman et al., from an information-theoretic lens. We demonstrate that differential privacy guarantees of iterative mappings can be determined by a direct application of contraction coefficients derived from strong data processing inequalities for $f$-divergences. In particular, by generalizing the Dobrushin's contraction coefficient for total variation distance to an $f$-divergence known as $\sE_\gamma$-divergence, we derive tighter bounds on the differential privacy parameters of the projected noisy stochastic gradient descent algorithm with hidden intermediate updates. 
	\end{abstract}

\section{Introduction and Motivation}
Differential privacy (DP) \cite{Dwork_Calibration, Dwork-OurData} has become the standard definition for designing privacy-preserving machine learning algorithms. One reason for its success is its \textit{operational} significance, which can be best described in terms of binary hypothesis testing (see, e.g., \cite{Wasserman, Kairouz_Composition}). Nevertheless, it is often difficult to compute DP guarantees for applications where a high number of data accesses is needed for a single analysis \cite{Abadi_MomentAccountant, Balle:Subsampling}. To obtain the DP parameters in such applications, which include machine learning models trained using stochastic gradient descent (SGD), one needs to resort to composition theorems which are often loose due to their generality. As a remedy, several variants of DP have been recently proposed \cite{RenyiDP,Concentrated_Dwork,ZeroDP,TruncatedDP} based on  R\'enyi divergence. These variants enjoy better composition properties. Among these variants, R\'enyi DP (RDP) has proven to be effective in studying private deep learning algorithms \cite{Abadi_MomentAccountant} especially when paired with \textit{sub-sampling} techniques \cite{Shiva_subsampling}. 

 Recently, the new framework of \textit{privacy amplification by iteration} was proposed by Feldman et al. \cite{Feldman2018PrivacyAB} as an alternative to privacy amplification by sub-sampling. This  framework possesses several advantages which makes it  well-suited for determining and enforcing privacy in distributed settings where data samples are stored locally by each user. Existing private algorithms based on sub-sampling require hiding the set of users participating in each update step of the model. This requirement, however, dictates either all data samples be stored centrally (i.e., no distributed setting) or all-to-all communication (i.e., excessive communication complexity). 
The new framework of privacy amplification by iteration relaxes these issues; it does not require the order of participating users to be random or hidden. On the other hand, it requires that all intermediate updates be hidden until a certain number of update steps are applied (e.g., not disclosing model update of SGD before a pre-specified step, say, $n$-th step).  

Since the intermediate updates are assumed to be hidden, one can view an iterative process as a concatenation of channels. To see this, let $\{\psi_t\}_{t=1}^n$ be a sequence of mapping and the update rule be given by \begin{equation}\label{Eq:Iterative_Vitaly}
    Y_{t} = \psi_{t}(Y_{t-1}) + Z_{t},
\end{equation} 
where $Y_0 = y_0\in \R^d$ and $\{Z_t\}_{t=1}^n$ are i.i.d.\ copies of a noise distribution $P_Z$. Let $\{Y'_t\}_{t=1}^n$ be the output of the same process started at $Y'_0= y'_0\in \R^d$. Letting $\mu_t$ and $\nu_t$ be the distributions of $Y_t$ and $Y_t'$, the \textit{strong} data processing inequality (SDPI) for $f$-divergences (see, e.g., \cite{Raginsky_SDPI, Anurak_SDPI}) implies that 
\begin{equation}\label{SDPI_f}
    D_f(\mu_n\|\nu_n)\leq D_f(\mu_1\|\nu_1)\prod_{t=1}^n\eta_f(\sK_t),
\end{equation}
where $D_f$ is an $f$-divergence and $\eta_f(\sK_t)$ is the \textit{contraction coefficient} (also known as strong data processing constant) of the Markov kernel $\sK_t(y) \coloneqq  P_{Y_t|Y_{t-1}=y} = P_{Z+\psi_t(y)}$ under $f$-divergence (see Section~\ref{Section:ContractionEgamma} for details). By exploiting the connection between DP and a certain $f$-divergence known as $\sE_\gamma$-divergence, we build upon \eqref{SDPI_f} to obtain bounds for DP parameters of iterative processes. Specifically, we study the noisy stochastic gradient descent algorithm and obtain tighter bounds for its DP parameters than those provided currently in the literature \cite{Feldman2018PrivacyAB,Balle2019mixing}.    
To do so, we obtain a closed-form expression for the contraction coefficient of Markov kernels under $\sE_\gamma$-divergence that generalizes the well-known Dobrushin's theorem \cite{Dobrushin}.

Our approach is inspired by the original work of Feldman et al.\ \cite{Feldman2018PrivacyAB}. They adopted RDP as the measure of privacy and proved the following SDPI result \cite[Theorem 1]{Feldman2018PrivacyAB} for the R\'enyi divergence of order $\alpha>1$: For the iterative process described in \eqref{Eq:Iterative_Vitaly} with $P_Z$ the Gaussian distribution $\N(0, \sigma^2\mathrm{I}_d)$,
\begin{equation}\label{SDPI_Renyi}
    D_\alpha(\mu_n\|\nu_n)\leq \frac{1}{n}D_\alpha(\mu_1\|\nu_1) = \frac{1}{n}\frac{\alpha \|y_0-y'_0\|}{2\sigma^2}.
\end{equation}
Despite its tractability, RDP lacks a clear operational interpretation. As a result, RDP guarantees are usually translated to DP guarantees via a transformation which is known to be loose, see, e.g.,  \cite[Proposition 3]{RenyiDP}. 

As a special case of iterative processes, we consider the \textit{noisy} SGD algorithm with Laplacian or Gaussian perturbation. Our empirical analyses show that the DP parameters of noisy SGD obtained by our approach are smaller than that of \cite{Feldman2018PrivacyAB, Balle2019mixing} (after applying the RDP to DP transformation). To capture common practice in machine learning applications, the input alphabet of the Markov kernels in this work are assumed to be compact. As a result, our analysis of contraction coefficients of such kernels is akin to the analysis of input-constrained channels performed by \cite{Yury_Dissipation}.

\section{Background}
\label{Sec:Background}

In this section, we briefly review privacy mechanisms, $f$-divergences and contraction coefficients. We also review a relation between DP and $\sE_\gamma$-divergence.

\subsection{Privacy Mechanisms}

The following examples describe two typical privacy mechanisms used in machine learning.

\begin{example}[(Private Queries)]
Let $\mathcal{X}$ be an arbitrary alphabet. A query is a function $f$ that takes a sample $\mathbb{D} \in \mathcal{X}^n$ and produces a {\it response} $y$ in the space of responses $\mathcal{Y}$. In this setting, a privacy mechanism $\sK$ takes a response $y\in\mathcal{Y}$ and produces another (random) response in the same space. In general, a privacy mechanism can be described by a Markov kernel $\sK:\mathcal{Y}\to\mathcal{P}(\mathcal{Y})$, i.e., a channel with the same input and output space $\Y$,  where $\mathcal{P}(\mathcal{Y})$ denotes the set of probability measures over $\mathcal{Y}$. Thus, the private query, say $\mathcal{M}$, is a random variable satisfying  $\mathcal{M}(\mathbb{D}) \sim \sK(f(\mathbb{D}))$.
\end{example}

\vspace{-5pt}

\begin{example}[(Stochastic Optimization)]
\label{Example:StochasticOptimization}
Let $\mathcal{Y}$ denote a parameter space, e.g., the coefficients in a linear regression model. Given a dataset $\mathbb{D} = \{x_1,\ldots,x_n\}\in \X^n$, typical stochastic optimization methods take an initial point $Y_0 \sim \mu_0 \in \mathcal{P}(\mathcal{Y})$ and further refine it through a random optimization process. The latter process typically depends on the dataset $\mathbb{D}$ and can be encoded by a Markov kernel $\sK_{\mathbb{D}}:\mathcal{Y}\to\mathcal{P}(\mathcal{Y})$. Furthermore, in some cases it is of iterative form, e.g., stochastic gradient descent, and the kernel $\sK_{\mathbb{D}}$ can be decomposed as $\sK_{\mathbb{D}} = \sK_{x_1} \cdots \sK_{x_n}$. Here, the randomness of the initial point and the optimization process may provide some level of privacy.
\end{example}

Motivated by the previous examples, we model privacy mechanisms as random mappings taking a data set $\mathbb{D}\in\mathcal{X}^n$ as input and producing an element in a given set $\mathcal{Y}$ as output. Furthermore, we assume that any privacy mechanism, say $\mathcal{M}$, is a random variable satisfying 
\begin{equation*}
    \mathcal{M}(\mathbb{D}) \sim \mu_0 \sK \coloneqq \int\mu_0(\text{d}y)\sK(y),
\end{equation*}
where the measure $\mu_0\in\mathcal{P}(\mathcal{Y})$ and the kernel $\sK:\mathcal{Y}\to\mathcal{P}(\mathcal{Y})$ may depend on $\mathbb{D}$, i.e., $\mu_0 = (\mu_0)_{\mathbb{D}}$ and $\sK = \sK_{\mathbb{D}}$.

\subsection{$f$-Divergence and Contraction Coefficients}

Given a convex function $f:(0,\infty)\to\mathbb{R}$ with $f(1)=0$, $f$-divergence between two probability measures $\mu$ and $\nu$ is defined in \cite{Ali1966AGC,Csiszar67} as
\begin{equation*}
    D_f(\mu\|\nu)\coloneqq \E_{\nu}\left[f\left(\frac{\text{d}\mu}{\text{d}\nu}\right)\right].
\end{equation*}

Let $\sK:\Y\mapsto\P(\Y)$ be a Markov kernel. Following the definition from Ahslwede and G\'acs \cite{ahlswede1976}, we define the \textit{contraction coefficient} (or strong data processing coefficient) of $\sK$ under $f$-divergence as
\begin{equation*}
    \eta_f(\sK)\coloneqq \sup_{\mu, \nu: \atop D_f(\mu\|\nu)\neq 0}\frac{D_f(\mu \sK\|\nu\sK)}{D_f(\mu\|\nu)}.
\end{equation*}
This quantity has been studied for several $f$-divergences, e.g., $\mathsf{KL}$-divergence for which $f(t) = t\log(t)$, $\chi^2$-divergence for which $f(t) = (t-1)^2$, and also total variation distance for which $f(t) = \frac{1}{2}|t-1|$. In particular, Dobrushin \cite{Dobrushin} showed that 
\begin{equation}
\label{eq:Dobrushin}
    \eta_\mathsf{TV}(\sK) = \sup_{y_1\neq y_2}\mathsf{TV}(\sK(y_1), \sK(y_2)),
\end{equation}
where $\mathsf{TV}(\mu,\nu)$ denotes the total variation distance between $\mu$ and $\nu$. It is worth noting that \eqref{eq:Dobrushin} has been extensively used in information theory \cite{Yury_Dissipation, Verdu:f_divergence}, statistics \cite{Dobrushin_Maxim} and graph theory \cite{PereBook2006, Moral2003OnCP}.
	
\subsection{Differential Privacy and $\sE_\gamma$-Divergence}

For an arbitrary alphabet $\X$, let $\X^n$ be the set of all datasets of size $n$. By definition, two datasets $\mathbb{D}$ and $\mathbb{D}'$ are \textit{neighboring}, denoted as $\mathbb{D} \sim \mathbb{D}'$, if their Hamming distance is equal to one. Given a randomized mechanism $\M$, we let $P_\mathbb{D}$ be the distribution of $\M(\mathbb{D})$, the output of $\M$ with $\mathbb{D}\in\X^n$ as the input. For $\eps\geq 0$ and $\delta\in [0,1]$, a mechanism $\M$ is said to be $(\eps, \delta)$-differentially private (DP) if
\begin{equation}
\label{eq:DefDP}
    P_{\mathbb{D}}(A)\leq e^\eps P_{\mathbb{D}'}(A) + \delta,
\end{equation}
for every measurable set $A\subset \Y$ and neighboring datasets $\mathbb{D} \sim \mathbb{D}'$. When $\delta=0$, we simply say that $\M$ is $\eps$-DP.

The definition of $(\eps,\delta)$-DP given in \eqref{eq:DefDP} can be reformulated in terms of $\sE_\gamma$-divergence, also known as hockey-stick divergence \cite{Yuri_finite, hockey_stick, Csiszar_Sheilds}. Given $\gamma\geq 1$, $\sE_\gamma$-divergence between two probability distributions $P$ and $Q$ is defined as
\begin{align}
	\sE_\gamma(P\|Q) &\coloneqq  \int_{\mathcal{Y}} [\text{d}(P-\gamma Q)(y)]_+ \label{eq:DefEgamma}\\
	&= \sup_{A\subset \Y} \left[P(A) - \gamma Q(A)\right]\nonumber\\
	&=  P\left(\imath_{P\|Q}>\log\gamma\right)-\gamma Q\left(\imath_{P\|Q}>\log\gamma\right), \label{eq:EgammaPDif}
\end{align}
where $[b]_+ = \max\{0,b\}$ and $\imath_{P\|Q}(t)\coloneqq \log\frac{\text{d}P}{\text{d}Q}(t)$ denotes the \textit{information density} between $P$ and $Q$. The $\sE_\gamma$-divergence is in fact an $f$-divergence associated with $f(t) = (t-\gamma)_+$ and it also satisfies that $\sE_1(P\|Q) = \mathsf{TV}(P,Q)$. The next theorem provides a relation between this divergence and $(\eps,\delta)$-DP.
\begin{figure*}
    \centering
    \includegraphics[height=0.13\textwidth,width=0.9\textwidth]{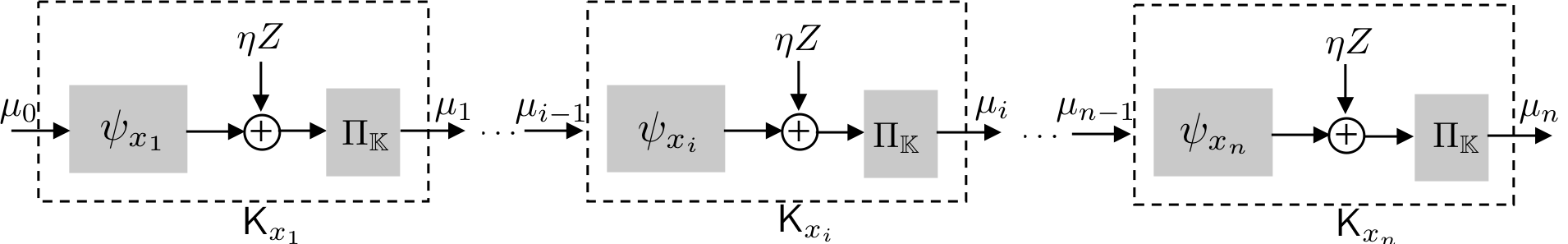}
    \caption{The schematic representation of the projected noisy stochastic gradient descent described algorithm in Algorithm~\ref{alg:PNSGD}. Given $\mu_0$ an arbitrary distribution on $\mathbb K$ and dataset $\mathbb D=\{x_1, \dots, x_n\}$, the $i$-th iteration is encoded by a projected additive kernel $\sK_{x_i}$ given by $y\mapsto \Pi_{\mathbb K}(\psi_{x_i}(y) + \eta Z_i)$ where $\psi_{x_i}(y) = y-\eta\nabla_y\ell(y, x_i)$. The output distribution of kernel $\sK_{x_i}$ is $\mu_i = \mu_0\sK_{x_1}\dots\sK_{x_{i-1}}$.}
    \label{fig:SDPI}
\end{figure*}
\begin{theorem}[\cite{Barthe:2013_Beyond_DP, Improving_Gaussian}]
\label{Thm:Balle}
	A mechanism $\M$ is $(\eps, \delta)$-DP if and only if, for all $\mathbb{D}\sim \mathbb{D}'$,
	\begin{equation*}
	    \mathsf{E}_{e^\eps}(P_{\mathbb{D}}\|P_{\mathbb{D}'})\leq \delta.
    \end{equation*}
\end{theorem}

By relating DP to $\sE_\gamma$-divergence,  this theorem enables us to invoke the SDPI relationship \eqref{SDPI_f}, specialized to $\sE_\gamma$-divergence, to obtain the DP parameters $\eps$ and $\delta$ of iterative processes.  To do so, we first need to compute the contraction coefficient under $\sE_\gamma$-divergence, which is addressed in the next section.

\section{Contraction of $\sE_{\gamma}$-Divergence}
\label{Section:ContractionEgamma}
In this section we establish a closed-form expression for the contraction coefficient of kernels under $\sE_\gamma$-divergence that generalizes the Dobrushin's theorem in \eqref{eq:Dobrushin}. We then instantiate this expression to introduce a family of practically-appealing kernels  with compact input alphabet. For ease of notation, we let $\eta_\gamma(\sK) \coloneqq \eta_{\sE_\gamma}(\sK)$.

\begin{theorem}
\label{Thm:Contraction_EGamma}
For any $\gamma\geq1$, we have 
\begin{equation}
\label{eq:DobrushinEgamma}
    \eta_\gamma(\sK)= \sup_{y_1,y_2\in\mathcal{Y}}\sE_{\gamma}(\sK(y_1)\|\sK(y_2)).
\end{equation}
\end{theorem}
Note that the Dobrushin's theorem \cite{Dobrushin} given in \eqref{eq:Dobrushin} corresponds to the special case of $\gamma = 1$ in Theorem~\ref{Thm:Contraction_EGamma}. This theorem implies that in order to compute the contraction coefficient of a Markov kernel $\sK$ under $\sE_\gamma$-divergence, one needs to compute $\sE_\gamma$-divergence between $\sK(y_1)$ and $\sK(y_2)$ for any $y_1,y_2\in\mathcal{Y}$. The following lemmas are useful for such task. For $m\in\mathbb{R}$ and $v>0$, let $\mathcal{L}(m,v)$ denote the Laplace distribution with mean $m$ and variance $2v^2$.

\begin{lemma}
\label{Lemma:EgammaLaplace}
For $m_1, m_2\in \R$ and $v>0$, we have
\begin{equation}
\label{Exa_Contraction_Laplacian}
    \sE_\gamma(\mathcal{L}(m_1,v)\|\mathcal{L}(m_2,v)) =\left[1- e^{\frac{v\log(\gamma) -|m_1-m_2|}{2v}}\right]_+.
\end{equation}
\end{lemma}
For $m\in\mathbb{R}^d$ and $\sigma>0$, let $\mathcal{N}(m,\sigma^2{\bf I}_d)$ denote the multivariate Gaussian distribution with mean $m$ and covariance matrix $\sigma^2{\bf I}_d$.
\begin{lemma}
\label{Lemma:EgammaGaussian}
For $\mathcal{N}_i = \mathcal{N}(m_i,\sigma^2{\bf I}_d)$, $i=1,2$, we have
\begin{equation*}
    \sE_{\gamma}(\N_1\| \N_2) = \sQ\left(\frac{\log\gamma}{\beta} - \frac{\beta}{2}\right) - \gamma \sQ\left(\frac{\log\gamma}{\beta} + \frac{\beta}{2}\right),
\end{equation*}
where  $\sQ(t) = \frac{1}{\sqrt{2\pi}} \int_t^\infty e^{-u^2/2}\text{d}u$ and $\beta = \frac{\|m_2-m_1\|}{\sigma}$. 
\end{lemma}

The previous lemma motivates the following definition.

\begin{definition}\label{Def:theta}
For $\gamma\geq1$, we define $\theta_\gamma:[0,\infty)\to [0,1]$ by
\begin{align}
    \theta_\gamma(r) &\coloneqq \sE_\gamma\left(\mathcal{N}(ru,{\bf I}_d) \| \mathcal{N}(0,{\bf I}_d)\right) \nonumber\\
    &= \sQ\left(\frac{\log\gamma}{r} 
- \frac{r}{2}\right) - \gamma \sQ\left(\frac{\log\gamma}{r} + \frac{r}{2}\right), \label{E_Gamma_Gaussian}
\end{align}
where $u\in\mathbb{R}^d$ is any vector of unit norm.
\end{definition}
With this definition at hand, we can write  
\begin{equation}
\label{theta_Gaussian}
    \sE_{\gamma}(\N(m_1, \sigma^2\mathbf{I}_d)\| \N(m_2, \sigma^2\mathbf{I}_d)) = \theta_\gamma\left(\frac{\|m_2 - m_1\|}{\sigma}\right).
\end{equation}
It is worth pointing out that $\theta_\gamma$ has a similar role as the function $R_\alpha$ introduced by Feldman {\it et al.} in \cite{Feldman2018PrivacyAB}.

The additive Gaussian kernel $\sK:\mathbb{R}^d\to\mathbb{R}^d$ is the kernel determined by $\sK(y) = \mathcal{N}(y,\sigma^2 {\bf I}_d)$ for some $\sigma>0$. This kernel models the privacy mechanisms which map $y \mapsto y + Z$ with $Z \sim \mathcal{N}({\bf 0},\sigma^2{\bf I}_d)$. An application of Theorem~\ref{Thm:Contraction_EGamma} and Lemma~\ref{Lemma:EgammaGaussian} shows that, under $\sE_\gamma$-divergence, the contraction coefficient of the additive Gaussian kernel is trivial\footnote{This is not surprising given the facts that $\eta_{\mathsf{TV}}(\sK) = 1$ for any Gaussian channels $\sK$ without input constraints \cite{Yury_Dissipation} and $\eta_{\mathsf{TV}}(\sK)=1$ if and only if $\eta_{f}(\sK)=1$ for all non-linear functions $f$ \cite{cohen1998comparisons}.}, i.e., $\eta_\gamma(\sK) = 1$. A similar conclusion holds for the additive Laplace kernel\footnote{The Euclidean norm of a $d$-dimensional Laplace noise vector is of order $d\log d$, see, e.g., \cite[Thm.~2]{SGD_Exponential_VS_Gaussian}. This asymptotic behavior makes Laplacian noise vectors highly inefficient for privacy purposes in the high dimensional setting. Therefore, in this paper we focus on the 1-dimensional case.} which is determined by $\sK(y) = \mathcal{L}(y,v)$ for some $v>0$.

Fortunately, the input and output of kernels appearing in applications tend to be bounded. Think, for example, of the kernel which models the update of the weights of a neural networks during its training. In this case, the weights are bounded either by design or by regularization mechanisms. Motivated by this observation, we say that a kernel $\sK:\mathbb{K}\to\mathbb{K}$ is the projected additive Gaussian kernel if it models the mechanism which maps $y \mapsto \Pi_\mathbb{K}(y+Z)$ where $\mathbb{K} \subset \mathbb{R}^d$ is compact and convex, $\Pi_{\mathbb{K}}$ is the projection operator onto $\mathbb{K}$ and $Z \sim \mathcal{N}(0,\sigma^2{\bf I}_d)$ for some $\sigma>0$. Similarly, we say that a kernel $\sK:\mathbb{K}\to\mathbb{K}$ is the projected additive Laplace kernel if it models the mechanism which maps $y \mapsto \Pi_\mathbb{K}(y+L)$ where $L \sim \mathcal{L}(0,v)$ for some $v>0$. These construction of kernels  will be instrumental in the analysis of privacy guarantee of iterative processes in the next section.

\section{Analysis of Iterative Mechanisms via Contraction Coefficients}

In this section, we consider iterative processes that can be decomposed into projected additive kernels. This constraint allows us to analyze the evolution of such processes through the lens of contraction coefficients. 

\subsection{General Setting}

Recall the stochastic optimization setting in Example\ \ref{Example:StochasticOptimization}, where $\mathcal{Y}\subset\mathbb{R}^d$ is a parameter space and $\mathbb{D} = \{x_1,\ldots,x_n\}$ is a dataset. In this context, an iterative stochastic optimization method is fully characterized by a set of kernels $\{\sK_x : x\in\mathcal{X}\}$ with $\sK_x:\mathcal{Y}\to\mathcal{P}(\mathcal{Y})$. The following lemma provides an upper bound for the $f$-divergence between the parameters returned when using two neighboring datasets.

\begin{lemma}
\label{Lemma:DPIIterative}
Let $\mu_0\in\mathcal{P}(\mathcal{Y})$ and $\{\sK_x : x\in\mathcal{X}\}$ be a family of kernels over $\mathcal{Y}$. If $\mathbb{D} = \{x_1,\ldots,x_n\}$ and $\mathbb{D}' = \{x'_1,\ldots,x'_n\}$ are neighbouring datasets with $x_i \neq x'_i$ for some $i\in[n]$, then
\begin{align*}
      D_f(\mu_{0}\sK_{x_1}\cdots \sK_{x_n}\|\mu_{0}\sK_{x'_1}\cdots \sK_{x_n'})\leq D_f(\mu_{i-1}\sK_{x_i}\|\mu_{i-1}\sK_{x'_i})\prod_{t=i+1}^n\eta_f(\sK_{x_t}),
\end{align*}
where $\mu_{i-1} \coloneqq \mu_0 \sK_{x_1} \cdots \sK_{x_{i-1}}$.
\end{lemma}

While the previous lemma follows from a routine application of the strong data processing inequality, it provides a natural framework to study the privacy guarantees of iterative optimization methods. In the following, we use it to study the privacy properties of the projected noisy stochastic gradient descent algorithm and some of its variations.

\subsection{Projected Noisy Stochastic Gradient Descent}
\label{Sec:PNSGD}


We now apply Lemma~\ref{Lemma:DPIIterative} to study the projected noisy stochastic gradient descent (PNSGD) algorithm under two different noise densities: Laplacian and Gaussian.


Assume that $\mathbb{K}$ is a compact and convex subset of $\mathbb{R}^d$ and that $\ell:\mathbb{K}\times\X\to\R_+$ is a loss function differentiable in its first argument. 
In the literature it is customary to assume regularity conditions on the loss function \cite{Feldman2018PrivacyAB, Balle2019mixing}. We make the following assumptions on the loss function: 
\begin{itemize}
	\item $y\mapsto \ell(y, x)$ is $L$-Lipschitz for all $x\in \X$,
	\item $y\mapsto \nabla_y \ell(y, x)$ is $\beta$-Lipschitz for all $x\in \X$,
	\item $y\mapsto \ell(y, x)$ is $\rho$-strongly convex for all $x\in \X$.
\end{itemize}
For a given a dataset $\mathbb{D} = \{x_1,\ldots,x_n\}$, the  PNSGD algorithm starts from a given point $Y_0\sim \mu_0\in \P(\mathbb K)$ (an arbitrary initial distribution) and then updates it according to the rule
\begin{equation}
\label{SGD_Update_Rulde}
    Y_{t+1} = \Pi_{\mathbb K}(Y_t - \eta[\nabla_y\ell(Y_t, x_{t+1})+Z_{t+1}]),
\end{equation} 
where $\Pi_{\mathbb K}:\R^d\to \mathbb{K}$ is the projection operator onto $\mathbb{K}$, $\eta>0$ is the learning rate and $\{Z_t\}_{t=1}^n$ is a collection of i.i.d.\ noise variables sampled from a distribution absolutely continuous w.r.t.\ the Lebesgue measure.  The PNSGD algorithm is summarized in Algorithm~\ref{alg:PNSGD}.

\begin{algorithm}
\caption{PNSGD Algorithm}
\label{alg:PNSGD}
\begin{algorithmic}
\small
\REQUIRE{Dataset $\mathbb{D} = \{x_1,\ldots,x_n\}$, learning rate $\eta>0$, initial point $Y_0 \sim \mu_0 \in \mathcal{P}(\mathbb{K})$ and i.i.d.\ copies $\{Z_t\}_{t=1}^n$ of a r.v.\ $Z$}
\FOR{$t\in\{0,\ldots,n-1\}$}
\STATE{$Y_{t+1} = \Pi_{\mathbb{K}}(Y_t - \eta [\nabla_y \ell(Y_t,x_{t+1}) + Z_{t+1}])$}
\ENDFOR
\RETURN{$Y_n$} 
\end{algorithmic}
\end{algorithm}

For any $x\in \X$, let  $\psi_{x}(y) \coloneqq y - \eta \nabla_y \ell(y,x)$. Notice that $y\mapsto \Pi_{\mathbb K}(\psi_{x}(y)+ \eta Z)$ is encoded by the projected additive Laplacian (resp., Gaussian) kernel if $Z$ is Laplacian (resp., Gaussian) noise variable. Given the dataset $\mathbb D=\{x_1, \dots, x_n\}$, one can therefore view the $i$-th iteration of the PNSGD algorithm  as a projected kernel   $\sK_{x_i}:\mathbb{K}\to\mathcal{P}(\mathbb{K})$ that models the mapping
\begin{equation*}
    y \mapsto \Pi_{\mathbb K}(\psi_{x_i}(y)+ \eta Z),
\end{equation*}
where $Z$ is the common distribution of $\{Z_t\}_{t=1}^n$. If $Y_1,\ldots,Y_n$ are produced by the PNSGD algorithm with $Y_0 \sim \mu_0$, then, for all $t\in[n]$, we have
\begin{equation*}
    Y_t \sim \mu_t = \mu_0 \sK_{x_1} \cdots \sK_{x_t}.
\end{equation*}
This allows us to express the PNSGD algorithm as a concatenation of $n$ channels, as illustrated in Fig.~\ref{fig:SDPI}.

Before delving into the privacy analysis of PNSGD, it is important to pause and adapt the definition of differential privacy to PNSGD setting. We recall from \cite{Feldman2018PrivacyAB} that a mechanism $\M$ is $(\eps, \delta)$-DP for its $i$th input if $\sEs(P_{\mathbb D}\|P_{\mathbb D'})\leq \delta$ for any pair of datasets $\mathbb D$ and $\mathbb D'$ differing on the $i$-th coordinate. Specializing Lemma~\ref{Lemma:DPIIterative} to $\sE_\gamma$-divergence, we can say that the PNSGD algorithm is $(\eps, \delta)$-DP for its $i$-th input if 
\begin{align}\label{Eq:EGamma_PNSGD}
     \sE_{e^\eps}(\mu_{0}\sK_{x_1}\cdots \sK_{x_n}\|\mu_{0}\sK_{x'_1}\cdots \sK_{x_n'})\leq \sE_{e^\eps}(\zeta_{x_i}\|\omega_{x_{i}})\prod_{t=i+1}^n\eta_{e^\eps}(\sK_{x_t}),
\end{align}
where $\zeta_{x_i}\coloneqq \mu_{i-1}\sK_{x_i}$ and $\omega_{x_i}\coloneqq \mu_{i-1}\sK_{x'_i}$.
Assuming $Z$ is either Laplacian or Gaussian, we can compute $\eta_{e^\eps}(\sK_{x_t})$.

\subsection{Laplacian Projected Noisy Stochastic Gradient Descent}
Here, we consider the PNSGD algorithm with Laplacian noise; i.e., $Z\sim \L(0, v)$ for some $v>0$. The following theorem establishes the $\eps$-DP property of such algorithm. For $L, \beta,$ and $\rho$ given in Section~\ref{Sec:PNSGD}, define
\begin{equation}\label{Def:M}
    M\coloneqq \sqrt{1-\frac{2\eta\beta\rho}{\beta+\rho}.}
\end{equation}

\begin{theorem}
\label{Thm:Laplacian_DP}
Assume that $\mathbb{K} = [a,b]$ for some $a<b$. Then PNSGD algorithm with Laplace noise is $(\eps,\delta)$-DP for its $i$-th input where $\eps\geq 0$ and 
$$\delta = \left(1-e^{\frac{\eps}{2} - \frac{L}{v}}\right)_+\left(1-e^{\frac{\eps}{2} - \frac{M(b-a)}{2\eta v}}\right)_+^{n-i}.$$
Consequently, we have $\delta =0$ if $\eps \geq  \min\{\frac{2L}{v},\frac{M(b-a)}{\eta v}\}$.
\end{theorem}

The use of Laplacian noise to provide $\eps$-DP for SGD algorithms was extensively studied, see e.g., \cite{Bassily_ERM, SGD_Exponential_VS_Gaussian, chaudhuri2011differentially}. 
Unlike previous results, Theorem~\ref{Thm:Laplacian_DP} is the first result regarding the privacy guarantees of PNSGD with Laplacian noise in the distributed-oriented framework proposed by Feldman et al.\ \cite{Feldman2018PrivacyAB}.  It is worth pointing out that our approach seems conceptually simpler than the approaches employed in  \cite{Feldman2018PrivacyAB, Balle2019mixing}. 

\subsection{Gaussian Projected Noisy Stochastic Gradient Descent}
Next, we assume that the noise distribution in the PNSGD algorithm is Gaussian, i.e., $Z\sim \N(0, \sigma^2\mathbf{I}_d)$. 

\begin{figure}
\hfill
\subfigure{\includegraphics[height =4cm, width=5.7cm]{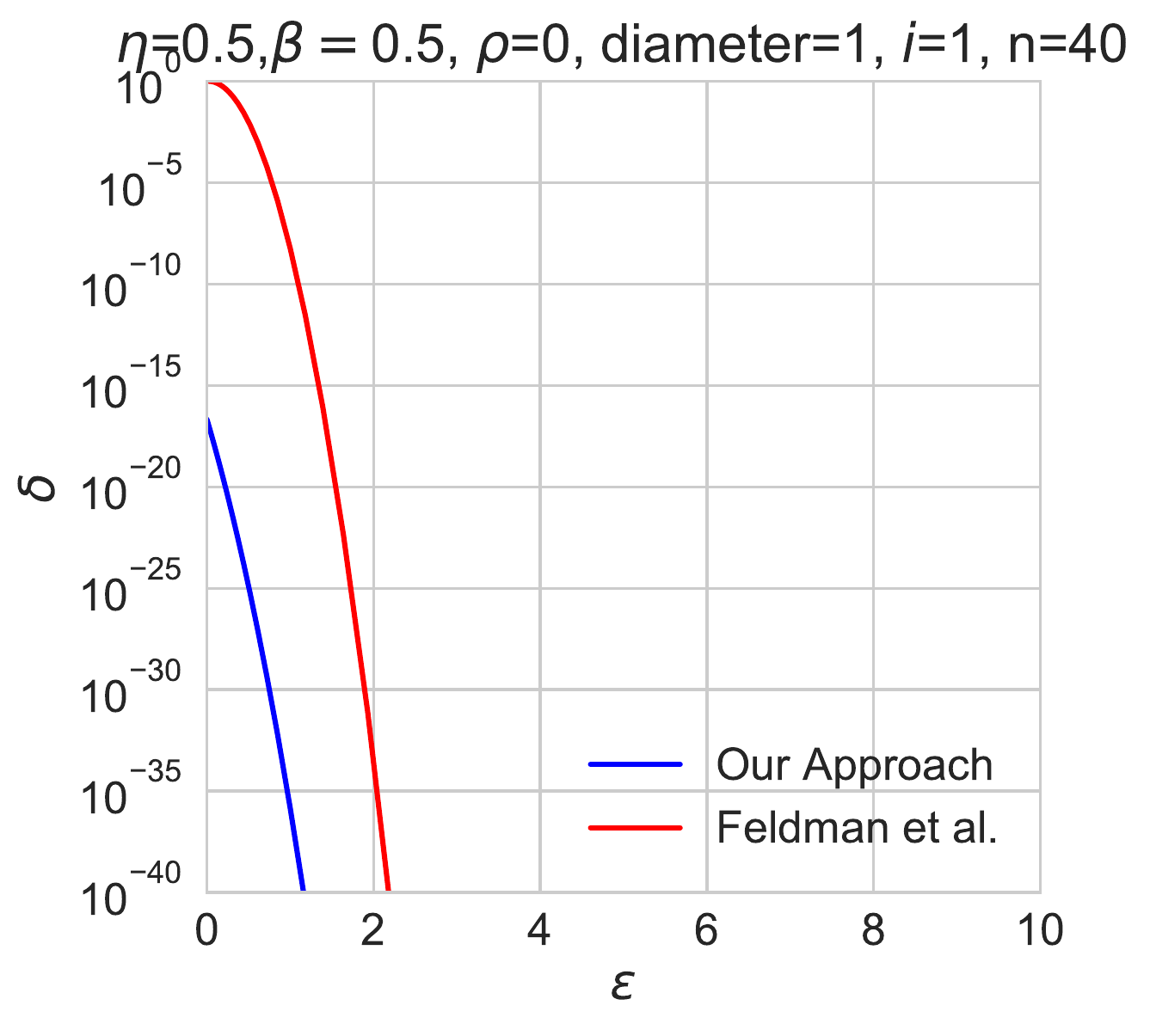}}
\hfill
\subfigure{\includegraphics[height =4cm,width=5.7cm]{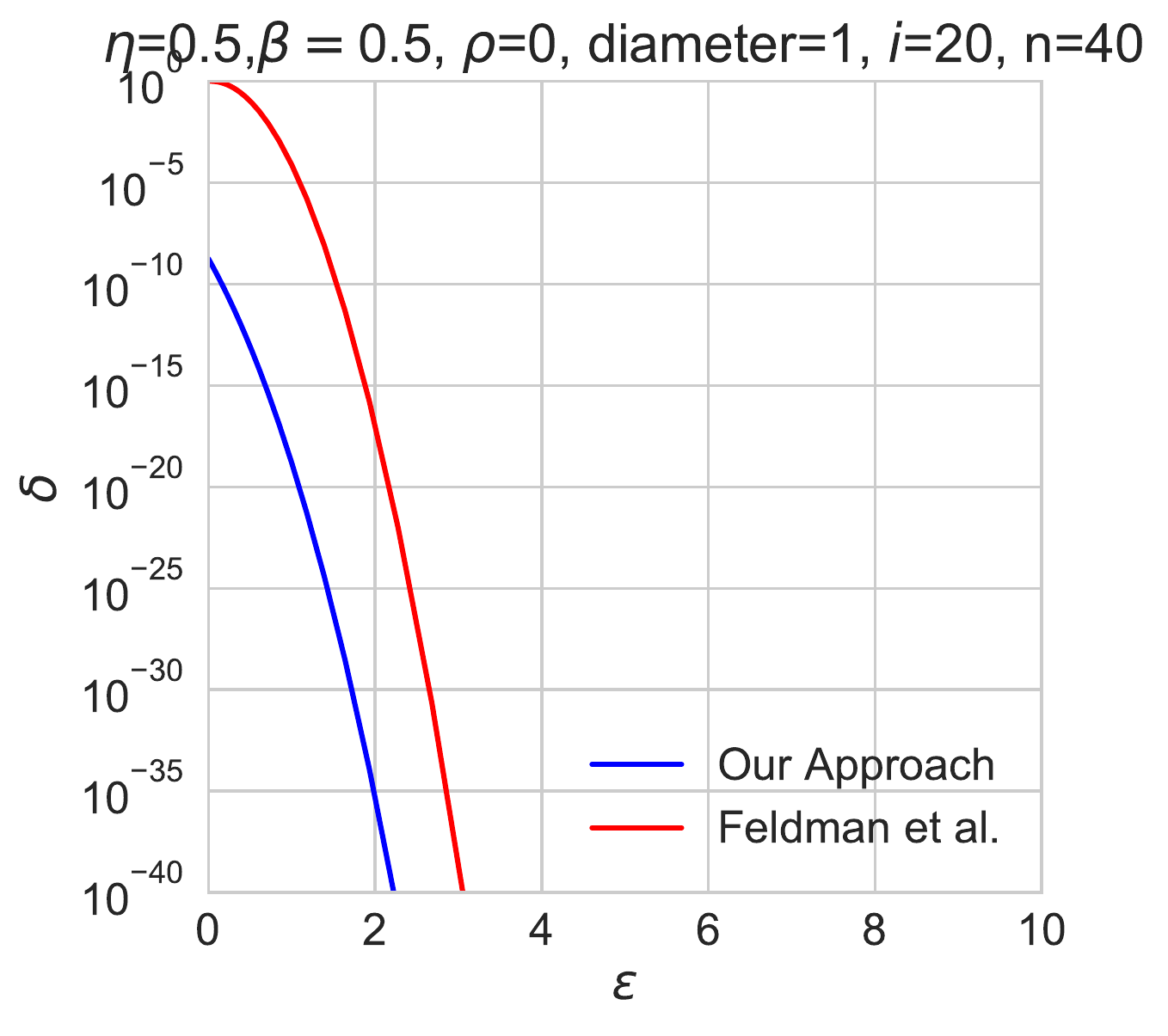}}
\hfill
\subfigure{\includegraphics[height =4cm,width=5.7cm]{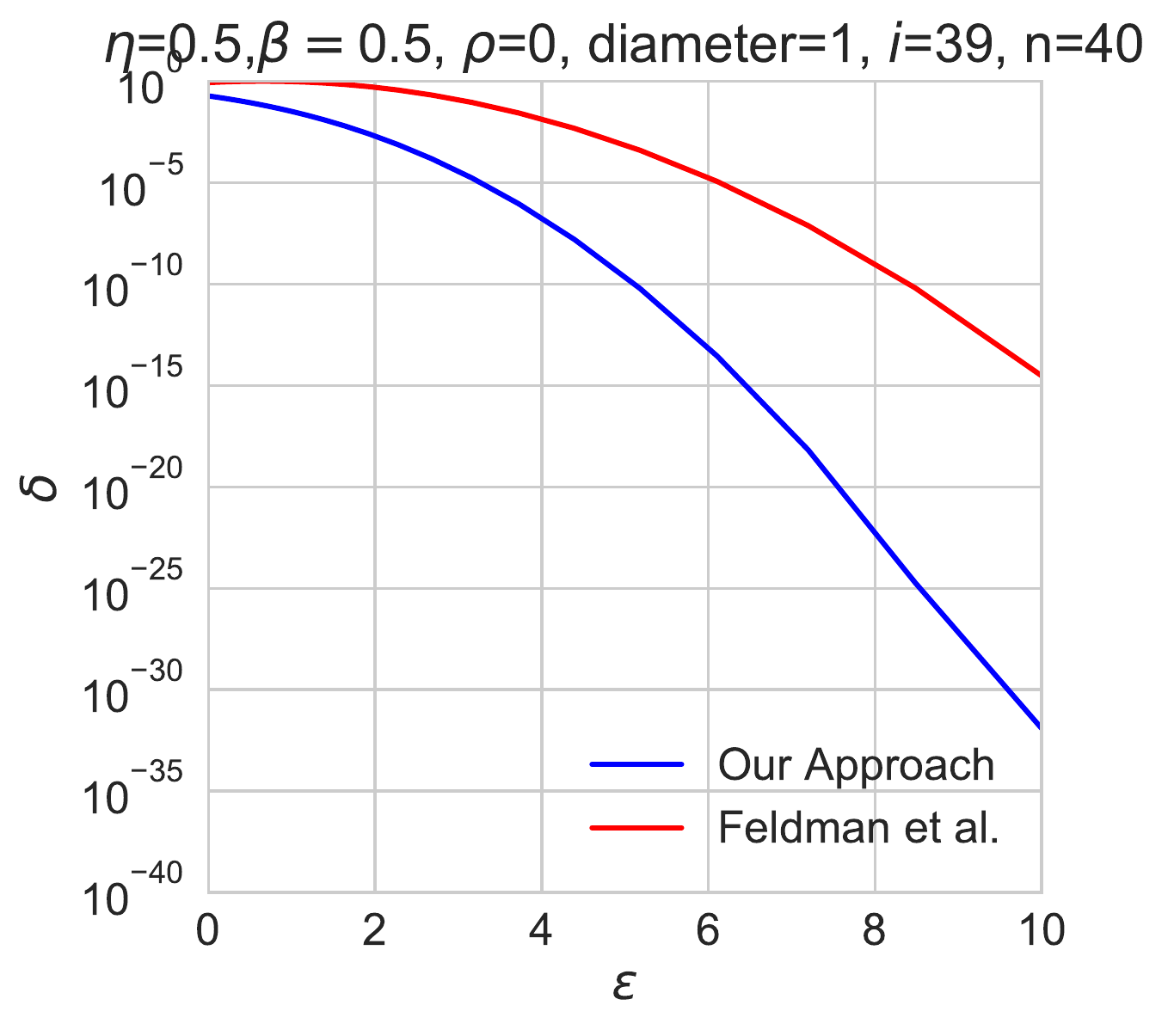}}
\hfill
\caption{The privacy parameters $\eps$ and $\delta$ of PNSGD with Gaussian noise having $\sigma =2$ and loss function with parameter $L=1$ and $\beta = 0.5$, computed using both Theorem~\ref{Thm:Gaussian_DP} and Balle et al. \cite[Theorem 5]{Balle2019mixing} for $i=1$, $i=20$, and $i=39$ in dataset of size $n=40$. Other parameters are as follows: $\eta =0.5, \rho=0$, and $D_\mathbb{K}=1$.}
\label{Fig:Comparison_Individual}
\end{figure}

\begin{theorem}
\label{Thm:Gaussian_DP}
Let $\mathbb{K} \subset \mathbb{R}^d$ be a compact and convex set. The PNSGD algorithm with Gaussian noise is $(\eps,\delta)$-DP for its $i$-th input where $\eps\geq 0$ and 
$$\delta = \theta_{e^\eps}\left(\frac{2 L}{\sigma}\right)\theta_{e^\eps}\left(\frac{MD_\mathbb{K}}{\eta \sigma}\right)^{n-i}.$$
\end{theorem}
Compared to Laplacian, the Gaussian perturbation has a better utility in high-dimensional setting, as illustrated in  \cite{SGD_Exponential_VS_Gaussian}. Hence, it has extensively appeared in DP literature as a \textit{de facto} mechanism for providing privacy  guarantees in training deep learning models \cite{Abadi_MomentAccountant}. Gaussian distribution is, in particular, appealing in the case of RDP as the R\'enyi divergence between two Gaussian distributions has a simple form (as opposed to $\sE_\gamma$-divergence).  This intuition, among others, led Feldman et al. \cite{Feldman2018PrivacyAB} and Balle et al.\ \cite{Balle2019mixing}  to adopt RDP to examine the PNSGD algorithm with Gaussian noise in the framework of privacy amplification by iteration. While the former studied the problem for cases where $M=1$ (i.e., $\rho = 0$), the latter assumed $M<1$ (i.e., $\rho>0$) and derived strictly better bounds for RDP guarantees. In fact, \cite[Theorem 5]{Balle2019mixing} reduces to \cite[Theorem 23]{Feldman2018PrivacyAB} when $\rho=0$.  We wish to compare Theorem~\ref{Thm:Gaussian_DP} with these results with or without strong convexity. 
To do so, we first need to convert the RDP guarantee given in \cite[Theorem 5]{Balle2019mixing} to $(\eps, \delta)$-DP. This conversion is a standard practice in DP literature and follows from an straightforward application of \cite[Proposition 3]{RenyiDP}.   
\begin{proposition}[Adapted from \cite{Balle2019mixing}]\label{Prop:RDP_DP_Balle}
The PNSGD algorithm with Gaussian perturbation is $(\eps, \tilde \delta)$-DP for its $i$-th input where $\eps>\kappa$ and   
\begin{equation}\label{Delta_Balle}
    \tilde \delta = e^{-\frac{1}{4\kappa}(\eps-\kappa)^2},
\end{equation}
where $\kappa = \frac{2 L^2}{(n-i)\sigma^2}M^{(n-i+1)}$ if $i\in [n-1]$ and $\kappa = 2 \frac{L^2}{\sigma^2}$ if $i = n$.  
\end{proposition}

 Note that $\delta$ in Theorem~\ref{Thm:Gaussian_DP} is given in terms of $\sQ$ function and hence it is challenging to analytically compare $\delta$ with $\tilde\delta$. Nevertheless, we provide several numerical comparisons.  In Fig.~\ref{Fig:Comparison_Individual}, we compare $\delta$ in Theorem~\ref{Thm:Gaussian_DP} with $\tilde\delta$ in Proposition~\ref{Prop:RDP_DP_Balle} for the first ($i=1$), middle ($i=20$) and the second to last ($i=39$) individuals in a dataset of size $n=40$ and $\sigma = 2$ with the assumption that the loss function is not strictly convex (i.e., $\rho=0$). As clearly seen, our approach outperfoms  \cite{Feldman2018PrivacyAB} especially for the individuals whose records were processed later in the algorithm.   

In Fig.~\ref{Fig:Comparison_eta}, we focus on the effect of strong convexity parameter $\rho$ on the privacy guarantee. We again depict $\delta$ and $\tilde\delta$ for the second half of the dataset: $i=20$, $i=30$, and $i=39$ in a dataset of size $n=40$ and $\sigma = 1$. Here, we assume that the loss function is  strictly convex with parameter $\rho=0.4$. As observed in this case, Theorem~\ref{Thm:Gaussian_DP} provides better privacy in the high privacy region (i.e., small $\eps$) as well as for the individuals who appear later in the dataset for all privacy region. 
 
\begin{figure}
	\hfill
	\subfigure{\includegraphics[height =4cm, width=5.7cm]{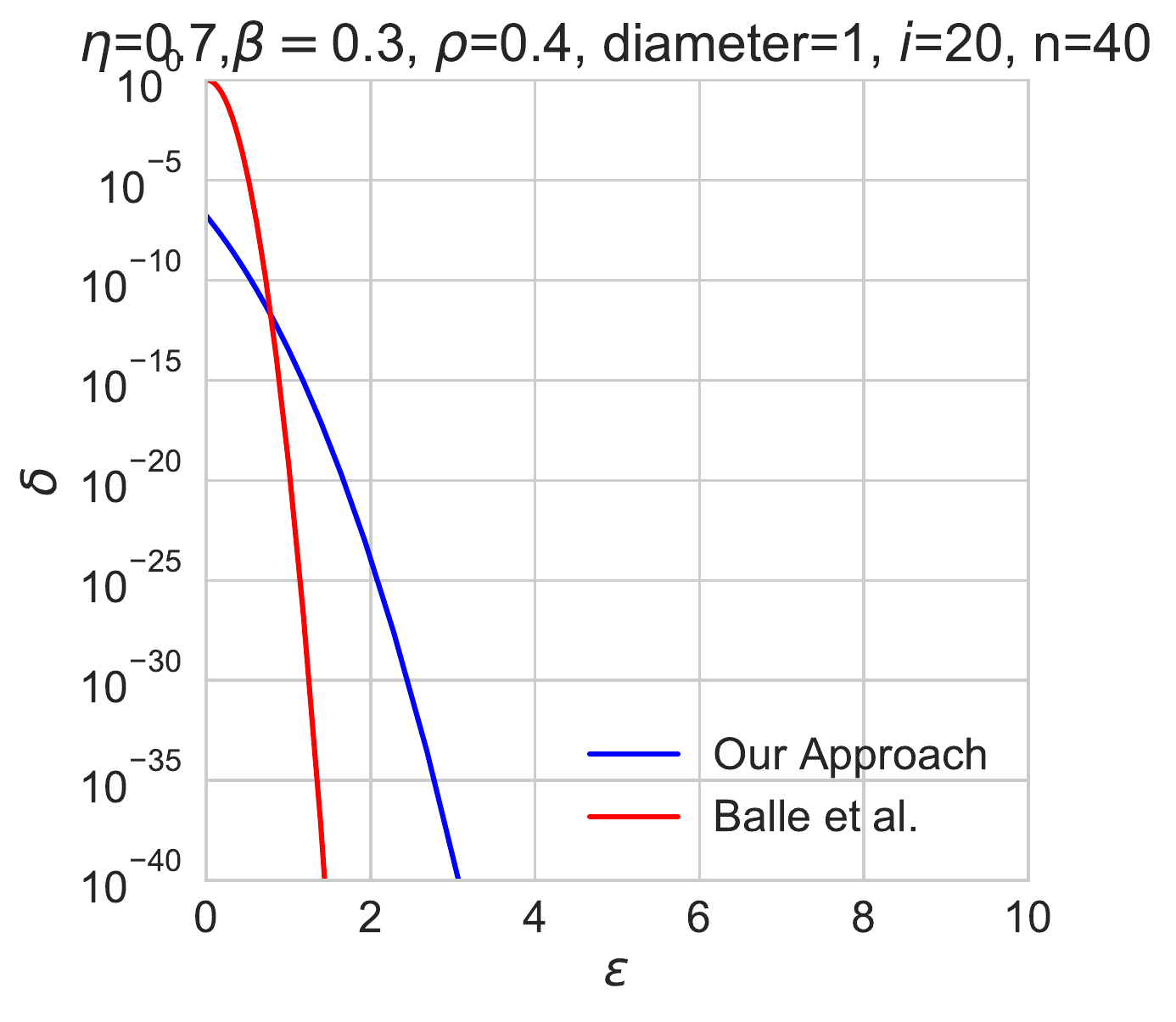}}
	\hfill
	\subfigure{\includegraphics[height =4cm,width=5.7cm]{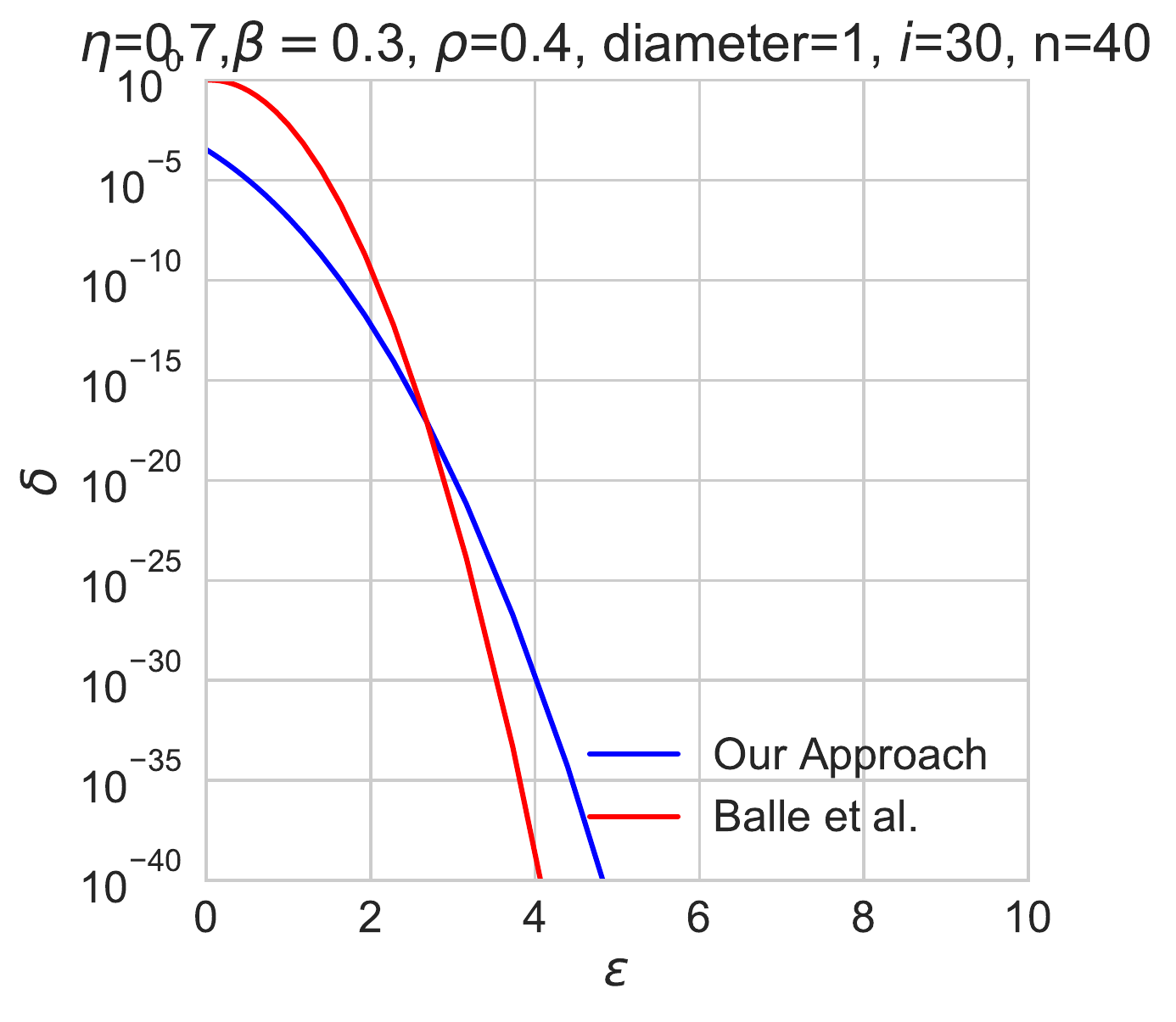}}
	\hfill
	\subfigure{\includegraphics[height =4cm,width=5.7cm]{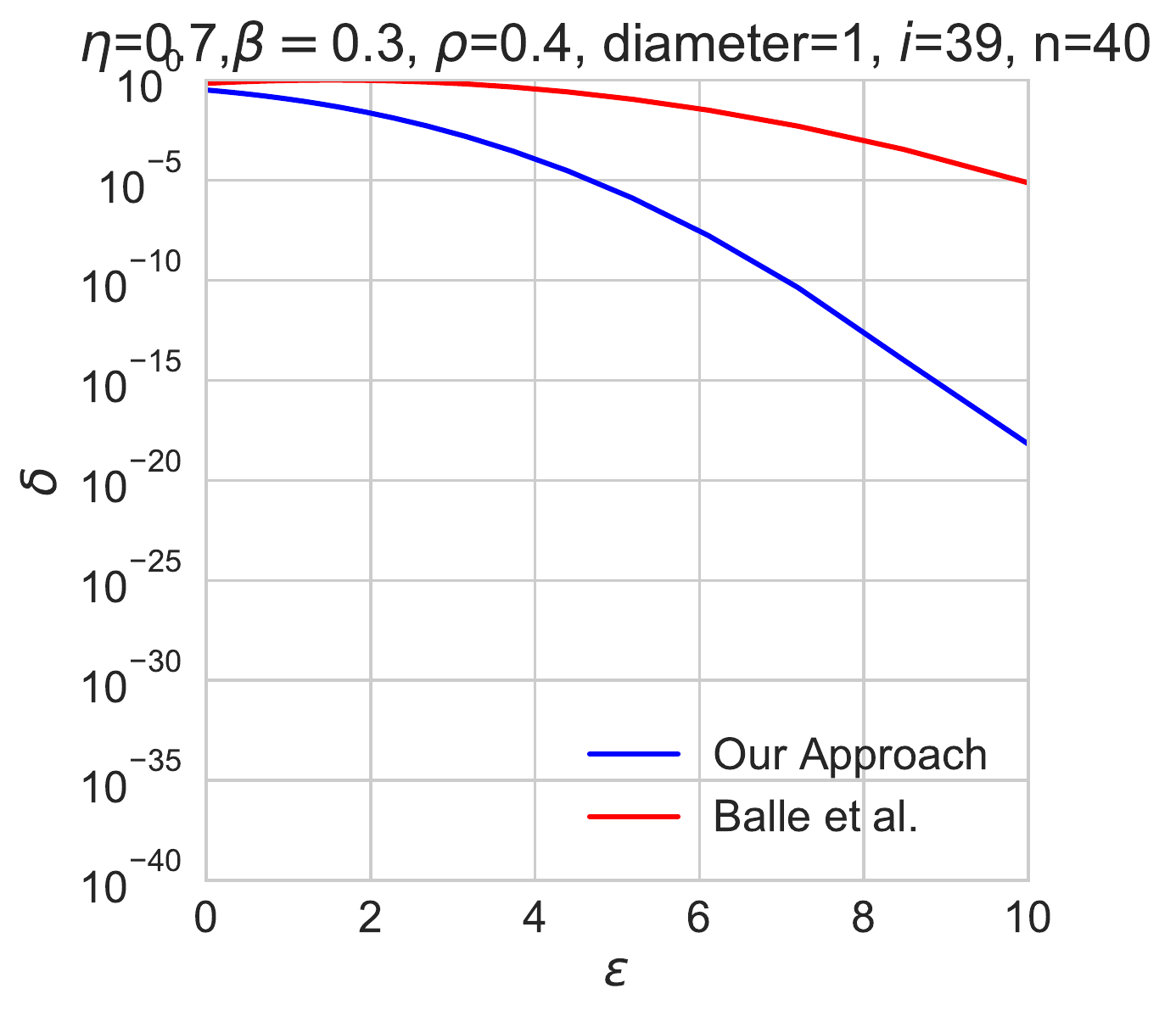}}
	\hfill
	\caption{The privacy parameters $\eps$ and $\delta$ of PNSGD with Gaussian noise having $\sigma=1$ and strongly convex loss function ($\rho =0.4$),  computed using both Theorem~\ref{Thm:Gaussian_DP} and Balle et al. \cite[Theorem 5]{Balle2019mixing} for $i=20$, $i=30$, and $i=39$ in a dataset of size $n=40$. Other parameters are as follows: $\eta =0.7, L=1, \beta =0.3$, and $D_\mathbb{K}=1$.}
	\label{Fig:Comparison_eta}
\end{figure}
\subsection{Randomly Stopped PNSGD Algorithm}
We end this section by pointing out a potential shortcoming of Theorems~\ref{Thm:Laplacian_DP} and \ref{Thm:Gaussian_DP} (and in general the privacy amplification by iteration framework): different individuals participating in the dataset experience different privacy guarantees; that is, individuals whose records were processed earlier experience higher privacy guarantee. This may not be justified in practice. To address this issue, we follow \cite{Feldman2018PrivacyAB} to consider the \textit{random stopping} for the PNSGD algorithm: namely, instead of iterating for $n$ steps, we pick a random time $T$ uniformly on $[n]$, stop the algorithm after $T$ steps and then output $Y_T$. The following theorem illuminates that such algorithm in fact \textit{uniformizes} the privacy guarantee among all individuals.   
\begin{theorem}\label{Thm:RandomTime}
Let $\mathbb{K} \subset \mathbb{R}^d$ be a compact and convex set. The  randomly-stopped PNSGD algorithm with Gaussian noise is $(\eps, \delta)$-DP with $\eps\geq 0$ and
\begin{equation}\label{delta_RandomTime}
    \delta = \frac{1}{n}\theta_{e^\eps}\left(\frac{2 L}{\sigma}\right)\left(1-\theta_{e^\eps}\left(\frac{MD_{\mathbb K}}{\eta \sigma}\right)\right)^{-1}.
\end{equation}
 \end{theorem}
The randomly stopped PNSGD was first proposed by Feldman et al.\ \cite{Feldman2018PrivacyAB} where they derived its RDP guarantee in \cite[Theorem 26]{Feldman2018PrivacyAB} \textit{only} if $\sigma$ satisfies a certain constraint. This constraint is due to the non-convexity of the map $(\mu, \nu)\mapsto D_\alpha(\mu\|\nu)$. In contrast, since $(\mu, \nu)\mapsto \sE_\gamma(\mu\|\nu)$ is jointly convex (as for any other $f$-divergences), Theorem~\ref{Thm:RandomTime} holds for any $\sigma$.  

Another approach to address the non-uniformity of privacy guarantees is to \textit{permute} the dataset first, via a random permutation and then feed it to the PNSGD algorithm. We will examine this approach in our future work.    

\small
\bibliographystyle{IEEEtran}
\bibliography{reference}

\normalsize
\appendix

\section{Deferred Proofs}

\begin{proof}[Proof of Theorem~\ref{Thm:Contraction_EGamma}]
Balle et al.\ \cite{Balle2019mixing} recently showed that, for $\delta\in(0,1)$,
\begin{equation}
\label{eq:Balle}
    \sup_{\mu, \nu:\atop \sE_\gamma(\mu\|\nu)\leq \delta}\frac{\sE_{\gamma}(\mu \sK\|\nu\sK)}{\sE_\gamma(\mu\|\nu)}\leq \sup_{y_1\neq y_2}\sE_{\gamma'}(\sK(y_1)\|\sK(y_2)),
\end{equation}
where $\gamma'\coloneqq 1+\frac{\gamma-1}{\delta}$. Now we show that the above inequality is indeed an equality. Let $y_1,y_2\in\mathcal{Y}$ be such that $y_1\neq y_2$. We define $\mu_\delta = \bar{\delta} \delta_{y_0}+ \delta \delta_{y_1}$ and $\nu_\delta =  (\bar{\delta}/\gamma) \delta_{y_0}+ (1-\bar{\delta}/\gamma)\delta_{y_2}$ where $\bar{\delta} \coloneqq 1 - \delta$ and $y_0\notin\{y_1,y_2\}$. In particular,
\begin{equation*}
    \sE_\gamma(\mu_\delta\| \nu_\delta) = \delta.
\end{equation*}
It is straightforward to verify that $\mu_\delta \sK = \bar{\delta} K(y_0) + \delta K(y_1)$ and $\nu_\delta \sK = (\bar{\delta}/\gamma) K(y_0) + (1-\bar{\delta}/\gamma)K(y_2)$. Hence, by \eqref{eq:DefEgamma},
\begin{align*}
	\sE_\gamma(\mu_\delta\sK\| \nu_\delta\sK) 
	&= \delta\int_{\Y}\left[\text{d}(\sK(y_1) - \gamma'\sK(y_2))(y)\right]_+\\
	&= \delta\sE_{\gamma'}(\sK(y_1)\|\sK(y_2)).
\end{align*}
Therefore, we obtain that
\begin{align*}
    \sup_{\mu,\nu:\atop \sE_\gamma(\mu\|\nu)\leq \delta}\frac{\sE_{\gamma}(\mu \sK\|\nu\sK)}{\sE_\gamma(\mu\|\nu)}
    &\geq \frac{\sE_{\gamma}(\mu_\delta \sK\|\nu_\delta\sK)}{\sE_\gamma(\mu_\delta\|\nu_\delta)}\\
    &= \sE_{\gamma'}(\sK(y_1)\|\sK(y_2)).
\end{align*}
Since $y_1$ and $y_2$ are arbitrary, we conclude that the reversed version of inequality \eqref{eq:Balle} holds true and hence
\begin{equation}
\label{eq:BalleEquality}
    \sup_{\mu,\nu:\atop \sE_\gamma(\mu\|\nu)\leq \delta}\frac{\sE_{\gamma}(\mu \sK\|\nu\sK)}{\sE_\gamma(\mu\|\nu)} = \sup_{y_1\neq y_2} \sE_{\gamma'}(\sK(y_1)\|\sK(y_2)),
\end{equation}
with $\gamma' = 1 + \frac{\gamma-1}{\delta}$. It is easy to verify that for fixed $\mu$ and $\nu$, $\gamma \mapsto \sE_\gamma(\mu||\nu)$ is continuous and decreasing. Hence, by taking the limit as $\delta\to1$ in \eqref{eq:BalleEquality}, the result follows.
\end{proof}	

\begin{proof}[Proof of Lemma~\ref{Lemma:EgammaLaplace}]
For ease of notation, let $\mathcal{L}_i=\mathcal{L}(m_i, v)$ for $i=1, 2$. It follows from the definition that 
\al{\sE_\gamma(\L_1\|\L_2)&= \frac{1}{2v}\int_{-\infty}^\infty \left[e^{-\frac{|t-m_1|}{v}} - e^{\frac{v\log\gamma-|t-m_2|}{v}}\right]_+\text{d}t\\
& = \frac{1}{2v}\int_{-\infty}^\infty \left[e^{-\frac{|t-\tilde m|}{v}} - e^{\frac{v\log\gamma-|t|}{v}}\right]_+\text{d}t
}
where $\tilde m = m_2-m_1$ (assuming $m_2\geq m_1$). Clearly, the above integral is non-zero only if $\tilde m\geq v\log\gamma$. With this in mind, we can write 
\al{\sE_\gamma(\L_1\|\L_2)&= \frac{1}{2v}\int_{\frac{\tilde m +v\log\gamma}{2}}^\infty \left[e^{-\frac{|x-\tilde m|}{v}} - e^{\frac{v\log\gamma-|x|}{v}}\right]\text{d}x\\
& = 1- e^{\frac{v\log\gamma -\tilde m}{2v}},}
and hence $$\sE_\gamma(\L_1\|\L_2) =\left[1- e^{\frac{v\log\gamma -\tilde m}{2v}}\right]_+.$$
The case $m_1\geq m_2$ is similar. In general, we can write
$$\sE_\gamma(\L_1\|\L_2) =\left[1- e^{\frac{v\log\gamma -|m_1-m_2|}{2v}}\right]_+.$$
\end{proof}

\begin{proof}[Proof of Lemma~\ref{Lemma:EgammaGaussian}]
Recall that $\mathcal{N}_i = \mathcal{N}(m_i,\sigma^2{\bf I}_d)$ with $i=1,2$. A direct computation shows that
\begin{equation*}
    \iota_{\mathcal{N}_1\|\mathcal{N}_2}(t) = \frac{2\langle t,m_1-m_2 \rangle + \|m_2\|^2 - \|m_1\|^2}{2\sigma^2}.
\end{equation*}
Thus, if we let $\beta=\frac{\|m_1-m_2\|}{\sigma}$, then
\begin{equation*}
    \iota_{\mathcal{N}_1\|\mathcal{N}_2}(Y) = \begin{cases} \N(0.5\beta^2, \beta^2) & Y\sim\N(m_1, \sigma^2\mathbf{I}_d),\\ \N(-0.5\beta^2, \beta^2) & Y\sim\N(m_2, \sigma^2\mathbf{I}_d).\end{cases}
\end{equation*}
Therefore, by expression for $\sE_\gamma$ given in \eqref{eq:EgammaPDif},
\begin{align}
    \sE_\gamma(\mathcal{N}_1\|\mathcal{N}_2) &= \sE_{\gamma}(\N(m_1, \sigma^2\mathbf{I}_d)\| \N(m_2, \sigma^2\mathbf{I}_d)) \nonumber\\
    &=\sQ\left(\frac{\log\gamma}{\beta} 
- \frac{1}{2}\beta\right) - \gamma \sQ\left(\frac{\log\gamma}{\beta} + \frac{1}{2}\beta\right), \nonumber
\end{align}
where $\sQ(t) = \Pr(\N(0,1)\geq t) =\int_t^\infty\frac{1}{\sqrt{2\pi}}e^{-u^2/2}\text{d}u$. 
\end{proof}

\begin{proof}[Proof of Lemma~\ref{Lemma:DPIIterative}]
Observe that $x_t = x_t'$ for all $t<i$ and $t>i$. In particular,
\begin{equation*}
    \mu_{i-1} \coloneqq \mu_0 \sK_{x_1} \cdots \sK_{x_{i-1}} = \mu_0 \sK_{x_1'} \cdots \sK_{x_{i-1}'}.
\end{equation*}
Let $\Delta = D_f(\mu_{0}\sK_{x_1}\cdots \sK_{x_n}\|\mu_{0}\sK_{x'_1}\cdots \sK_{x_n'})$. By an iterative application of the data processing inequality, we conclude that
\begin{align*}
    \Delta &= D_f(\mu_{i-1} \sK_{x_i} \sK_{x_{i+1}} \cdots \sK_{x_n}\| \mu_{i-1} \sK_{x_i'} \sK_{x_{i+1}} \cdots \sK_{x_n})\\
    &\leq D_f(\mu_{i-1}\sK_{x_i}\|\mu_{i-1}\sK_{x'_i})\prod_{t=i+1}^n\eta_f(\sK_{x_t}),
\end{align*}
as desired.
\end{proof}

\begin{proof}[Proof of Theorem~\ref{Thm:Laplacian_DP}] Let $\mathbb{D} = \{x_t\}_{t=1}^n$ and $\mathbb{D}' = \{x_t'\}_{t=1}^n$ be two neighbouring datasets with $x_i \neq x'_i$. By Theorem~\ref{Thm:Balle}, it is enough to show that
\begin{equation*}
    \Delta \coloneqq \sE_{e^\eps}(\mu_0 \sK_{x_1} \cdots \sK_{x_n} \| \mu_0 \sK_{x_1'} \cdots \sK_{x_n'}) \leq \delta,
\end{equation*}
where $\mu_0$ is the initial distribution of PNSGD and $\sK_{x}$ models the update rule in \eqref{SGD_Update_Rulde}. Following \eqref{Eq:EGamma_PNSGD}, we obtain that 
\begin{equation}\label{Lap0}
    \Delta \leq \sE_{e^\eps}(\zeta_{x_i}\|\omega_{x_{i}})\prod_{t=i+1}^n\eta_\gamma(\sK_{x_t})
\end{equation}
where $\sK_{x_t}$ is a projected additive Laplacian kernel,  $\zeta_i = \mu_0\sK_{x_1}\dots\sK_{x_{i}}$ and $\omega_i = \mu_0\sK_{x'_1}\dots\sK_{x'_{i}}.$ 

We begin by recalling Lemma~\ref{Lemma:EgammaLaplace} and the data processing inequality to bound $\eta_{e^\eps}(\sK_{x_t})$ as 
\begin{eqnarray*}
    \eta_{e^\eps}(\sK_{x_t}) &=& \sup_{y_1, y_2\in \mathbb K}\sE_{e^\eps}(\Pi_{\mathbb K}\mathsf{Lap}(\psi_{x_t}(y_1), \eta v)\|\Pi_{\mathbb K}\mathsf{Lap}(\psi_{x_t}(y_2), \eta v))\nonumber\\
    &\leq&\sup_{y_1, y_2\in \mathbb K}\sE_{e^\eps}(\mathsf{Lap}(\psi_{x_t}(y_1), \eta v)\|\mathsf{Lap}(\psi_{x_t}(y_2), \eta v))\\
&=&\sup_{y_1, y_2\in \mathbb K}\left[1-e^{\frac{\eps}{2} - \frac{|\psi_{x_t}(y_1)-\psi_{x_t}(y_2)|}{2\eta v}}\right]_+
\end{eqnarray*}
To refine the above bound for $\eta_{e^\eps}$, we resort to the following standard result in convex optimization, see, e.g., \cite{Nesterov:ILC} or \cite[Theorem 3.12]{Convex_Bubeck}. Notice that we say a function $f:\mathbb K\to \R^d$ is $\beta$-\textit{smooth} if $a\mapsto \nabla f(a)$ is $\beta$-Lipschitz.  
\begin{lemma}\label{Lemma_M}
	Let $\mathbb K\subset \R^d$ be a convex set and suppose $f:\mathbb K\to \R^d$ is $\beta$-smooth and $\rho$-strongly convex. If $\eta<\frac{2}{\beta+\rho}$, then the map $\psi(a) = a-\eta\nabla f(a)$ is $M$-Lipschitz on $\mathbb K$ with $M=\sqrt{1-\frac{2\eta\beta\rho}{\beta+\rho}}$.
\end{lemma}
In light of this lemma, we have $|\psi_{x_t}(y_1) -\psi_{x_t}(y_2)|\leq M |y_1-y_2|\leq M(b-a)$ for any $y_1, y_2\in \mathbb K$. Hence, we obtain from above 
\begin{equation}\label{Lap1}
    \eta_{e^\eps}(\sK_{x_t}) \leq \left[1-e^{\frac{\eps}{2} - \frac{M(b-a)}{2\eta v}}\right]_+
\end{equation}
On the other hand, we use Jensen's inequality to compute $\sE_\gamma(\zeta_i\|\omega_i)$:
\begin{align}
	\sEs(\zeta_{x_i}\|\omega_{x_i})&\leq \int \sE_{e^\eps}(\sK_{x_i}(a)\|\sK_{x'_i}(a))\mu_{i-1}(\text{d}a) \nonumber\\
	&\leq  \int \sEs(\mathsf{Lap}(\psi_{x_i}(a), \eta v)\|\mathsf{Lap}(\psi_{x'_i}(a), \eta v))\mu_{i-1}(\text{d}a) \nonumber\\
	& = \int\left[1-e^{\frac{\eps}{2}-\frac{|\psi_{x_i}(a)-\psi_{x'_i}(a)|}{2\eta v}}\right]_+ \mu_{i-1}(\text{d}a)\nonumber\\
	&\leq \left[1-e^{\frac{\eps}{2}-\frac{L}{v}}\right]_+\label{Lap2}
\end{align}	
where the first inequality is due to the Jensen's inequality (note that $(\mu, \nu)\mapsto \sE_\gamma(\mu\|\nu)$ is jointly convex for all $\gamma>1$ as for any other $f$-divergences) and  the last inequality comes from the fact that 
\begin{align*}
    |\psi_{x_i}(a)-\psi_{x'_i}(a)| &=\eta|\nabla_a\ell(a, x_i)-\nabla_a\ell(a, x'_i)|\\
    &\leq \eta[|\nabla_a\ell(a, x_i)|+|\nabla_a\ell(a, x'_i)|]\leq 2\eta L,
\end{align*}
where we use the fact that $y \mapsto \ell(y,x)$ is $L$-Lipschitz for all $x\in\mathcal{X}$. 
Plugging \eqref{Lap1} and \eqref{Lap2} into \eqref{Lap0}, we conclude the proof. 
\end{proof}

\begin{proof}[Proof of Theorem~\ref{Thm:Gaussian_DP}]
Let $\mathbb{D} = \{x_t\}_{t=1}^n$ and $\mathbb{D}' = \{x_t'\}_{t=1}^n$ be two neighbouring datasets with $x_i \neq x'_i$. By Theorem~\ref{Thm:Balle}, it is enough to show that
\begin{equation*}
    \Delta \coloneqq \sE_{e^\eps}(\mu_0 \sK_{x_1} \cdots \sK_{x_n} \| \mu_0 \sK_{x_1'} \cdots \sK_{x_n'}) \leq \delta,
\end{equation*}
where $\mu_0$ is the initial distribution of PNSGD and $\sK_{x}$ models the update rule in \eqref{SGD_Update_Rulde}. Following \eqref{Eq:EGamma_PNSGD}, we obtain that 
\begin{equation}\label{Gau0}
    \Delta \leq \sE_{e^\eps}(\zeta_{x_i}\|\omega_{x_{i}})\prod_{t=i+1}^n\eta_\gamma(\sK_{x_t}),
\end{equation}
where $\sK_{x_t}$ is a projected additive Gaussian kernel,  $\zeta_{x_i} = \mu_0\sK_{x_1}\dots\sK_{x_{i}} = \mu_{i-1}\sK_{x_i}$ and $\omega_{x_i} = \mu_0\sK_{x'_1}\dots\sK_{x'_{i}}= \mu_{i-1}\sK_{x'_i}.$ We begin by recalling Lemma~\ref{Lemma:EgammaGaussian} and data processing inequality to  bound $\eta_\gamma(\sK_{x_t})$, for $\gamma = e^\eps$, as
\begin{align}
\eta_\gamma(\sK_{x_t}) &= \sup_{y_1\neq y_2\in \mathbb K} \sE_\gamma(\Pi_{\mathbb K}\N(\psi_{x_t}(y_1), \eta^2 \sigma^2\mathbf{I})\|\Pi_{\mathbb K}\N(\psi_{x_t}(y_2), \eta^2 \sigma^2\mathbf{I}))\nonumber\\
&= \sup_{\|y_1-y_2\|\leq D_{\mathbb K}} \sE_\gamma(\Pi_{\mathbb K}\N(\psi_{x_t}(y_1), \eta^2 \sigma^2\mathbf{I})\|\Pi_{\mathbb K}\N(\psi_{x_t}(y_2), \eta^2 \sigma^2\mathbf{I}))\nonumber\\
&\leq \sup_{\|y_1-y_2\|\leq D_{\mathbb K}} \sE_\gamma(\N(\psi_{x_t}(y_1), \eta^2 \sigma^2\mathbf{I})\|\N(\psi_{x_t}(y_2), \eta^2 \sigma^2\mathbf{I}))\nonumber\\
&=\sup_{\|y_1-y_2\|\leq D_{\mathbb K}}  \theta_\gamma\left(\frac{\|\psi_{x_t}(y_1)-\psi_{x_t}(y_2)\|}{\eta\sigma}\right)\nonumber\\
&\leq   \theta_\gamma\left(\frac{MD_{\mathbb K}}{\eta\sigma}\right),\label{Compact1}
\end{align}
where the last inequality comes from Lemma~\ref{Lemma_M} and the fact that $r \mapsto \theta_\gamma(r)$ is increasing.

On the other hand, we use Jensen's inequality to bound $\sE_{\gamma}(\zeta_{x_i}\|\omega_{x_i})$ as 
\begin{eqnarray}
	\sE_\gamma(\zeta_i\|\omega_i)&\leq & \int \sE_\gamma(\sK_{x_i}(a)\|\sK_{x'_i}(a))\mu_{i-1}(\text{d}a) \nonumber\\
	&=& \int \theta_\gamma\left(\frac{\|\psi_{x_i}(a)-\psi_{x'_i}(a)\|}{\eta\sigma}\right)\mu_{i-1}(\text{d}a)\nonumber\\
	&\leq&\theta_\gamma\left(\frac{2 L}{\sigma}\right), \label{Compact_2}
	\end{eqnarray}
where the first inequality is due to the Jensen's inequality (note that $(\mu, \nu)\mapsto \sE_\gamma(\mu\|\nu)$ is jointly convex for all $\gamma>1$ as for any other $f$-divergences) and the last inequality and the second inequality stems from the following 
\begin{align*}
    \|\psi_{x_i}(a)-\psi_{x'_i}(a)\| &=\eta\|\nabla\ell(a, x_i)-\nabla\ell(a, x'_i)\|\\
    &\leq \eta[\|\nabla\ell(a, x_i)\|+\|\nabla\ell(a, x'_i)\|]\leq 2\eta L.
\end{align*}
Plugging \eqref{Compact1} and \eqref{Compact_2} into \eqref{Gau0}, we conclude the proof. 
\end{proof}

 \begin{proof}[Proof of Proposition~\ref{Prop:RDP_DP_Balle}]
First note that both \cite[Theorem 23]{Feldman2018PrivacyAB} and \cite[Theorem 5]{Balle2019mixing} quantifies the RDP guarantee of PNSGD. However, as the latter strictly improves the former under the assumption of strong convexity, we only consider the former and convert it to DP guarantee.  Given $\alpha>1$, a mechanism $\M$ is said to be $(\alpha, \lambda)$-RDP \cite{RenyiDP} if $D_\alpha(P_{\mathbb D}\|P_{\mathbb {D'}})\leq \lambda$ for all $\mathbb D\sim \mathbb D'$ where $P_{\mathbb D}$ is the output distribution of $\M$ running on the dataset $\mathbb D$. As before, this definition can be adapted as follows: $\M$ is said to be $(\alpha, \lambda)$-RDP for its $i$-th input if the above holds for dataset $\mathbb D$ and $\mathbb D'$ differing in the $i$-th coordinate. 
It is shown in \cite[Proposition 3]{RenyiDP} and \cite[Theorem 2]{Abadi_MomentAccountant} that 
\begin{equation}\label{RDP_to_DP}
    \M ~\text{is}~ (\alpha, \lambda)\text{-RDP} \Longrightarrow \M ~\text{is}~ (\eps, \delta(\eps, \alpha))\text{-DP}, 
\end{equation}
where 
$$\delta(\eps, \alpha) \coloneqq e^{-(\alpha-1)(\eps-\lambda)}.$$ Note that in general $\lambda$ is a function of $\alpha$, hence $\delta$ depends on $\eps$ and $\alpha$. 

Balle et al. \cite[Theorem 5]{Balle2019mixing} proved that  PNSGD is $(\alpha, \kappa\alpha)$-RDP for its $i$th input and $\alpha>1$ where $$\kappa  = \frac{2 L^2}{n-i}M^{(n-i+1)},$$
for $i\in[n-1]$ and $\kappa =2\alpha L^2$ for $i=n$, where $M$ is defined in Lemma~\ref{Lemma_M}. 
We use \eqref{RDP_to_DP} to obtain the \textit{best} DP $\delta$ parameter: 
PNSGD  is $(\eps, \tilde \delta)$-DP for   
\begin{equation}\label{Feldman_delta}
    \tilde \delta= \inf_{\alpha>1}e^{-(\alpha-1)(\eps-\kappa\alpha)}. 
\end{equation}
Solving this minimization problem, we obtain that the minimizer is $\alpha^* =\frac{\eps+\kappa}{2\kappa}$ and the minimum value is 
\begin{equation}
    \tilde \delta = e^{-\frac{1}{4\kappa}(\eps-\kappa)^2}.
\end{equation}
Notice that since $\alpha>1$, we must have $\eps>\kappa$. 
\end{proof}

\begin{proof}[Proof of Theorem~\ref{Thm:RandomTime}]
Let $T$ be a uniform random variable on the set $[n]$. We assume that the PNSGD algorithm stops at the random time $T$. Recall that we are given two dataset $\mathbb D=\{x_1, \dots, x_n\}$ and $\mathbb D'= \{x'_1, \dots, x'_n\}$  with $x_j=x'_j$ for all $j\in [n]$ except $j=i$. Let $\mu_T$ and $\nu_T$ be the output distribution of the randomly-stopped PNSGD algorithm running on $\mathbb D$ and $\mathbb D'$, respectively. We can write 
$$\mu_T = \frac{1}{n} \sum_{r=1}^n \mu_0 \sK_{x_1}\dots \sK_{x_r}, $$
and 
$$\nu_T = \frac{1}{n}\sum_{r=1}^{n} \mu_{0}\sK_{x'_1}\dots \sK_{x'_r}.$$
Hence, the 
convexity of $(P, Q)\mapsto \sE_\gamma(P\|Q)$ and Jensen's inequality imply that
\begin{equation*}
  \mathbb \sE_{e^\eps}(\mu_T\|\nu_T) \leq  \frac{1}{n} \sum_{r = 1}^n   \mathbb \sE_{e^\eps}( \mu_0 \sK_{x_1}\dots \sK_{x_r} \| \mu_0 \sK_{x_1'}\dots \sK_{x_r'}).
\end{equation*}
Recall that $x_j = x_j'$ for all $j\neq i$. In particular, $\mu_0 \sK_{x_1}\dots \sK_{x_r} = \mu_0 \sK_{x_1'}\dots \sK_{x_r'}$ for all $r<i$ and hence
\begin{equation*}
    \mathbb \sE_{e^\eps}(\mu_T\|\nu_T) \leq  \frac{1}{n} \sum_{r = i}^n   \mathbb \sE_{e^\eps}(\mu_{i-1}\sK_{x_i}\dots \sK_{x_r}\|\mu_{i-1}\sK_{x'_i}\dots \sK_{x'_r}).
\end{equation*}
Finally, by applying Theorem~\ref{Thm:Gaussian_DP},
\begin{align*}
  \mathbb \sE_{e^\eps}(\mu_T\|\nu_T) &\leq  \frac{\theta_{e^\eps}\left(\frac{2L}{\sigma}\right)}{n} \sum_{r = i}^n \theta_{e^\eps}\left(\frac{MD_{\mathbb K}}{\eta\sigma}\right)^{r-i}\\
  &=\frac{\theta_{e^\eps}\left(\frac{2L}{\sigma}\right)}{n} \frac{1-\theta_{e^\eps}\left(\frac{MD_{\mathbb K}}{\eta\sigma}\right)^{n-i+1}}{1-\theta_{e^\eps}\left(\frac{MD_{\mathbb K}}{\eta\sigma}\right)} \\
  &\leq  \frac{\theta_{e^\eps}\left(\frac{2L}{\sigma}\right)}{n} \left(1-\theta_{e^\eps}\left(\frac{MD_{\mathbb K}}{\eta\sigma}\right)\right)^{-1},
\end{align*}
as we wanted to prove.
\end{proof}

\end{document}